\documentclass{llncs}
\usepackage{enumerate}
\usepackage{graphicx}
\usepackage{url}
\usepackage{xspace}
\usepackage{amsmath}
\usepackage{amssymb}
\usepackage{subfigure}

\spnewtheorem*{oproof}{Outline of the proof}{\rm\bfseries}{\rmfamily}

\usepackage{tikz}
\usetikzlibrary{automata,arrows,shapes,decorations.pathreplacing}
\tikzstyle{every node}=[circle,draw=black,inner sep = -2mm,minimum size = 6mm,fill=white]

\newcounter{instr}
\newcommand{\ninstr}{\refstepcounter{instr}\theinstr.}

\newcommand{\AND}{\mathrel{\&}}
\newcommand{\OR}{\mathrel{|}}
\newcommand{\NOT}{\mathop{\sim}}
\newcommand{\MOD}{\mathrel{\text{mod}}} 
\newcommand{\bpstyle}[1]{\mathsf{#1}}
\newcommand{\gpattern}{S_1\cdot j_1\cdot S_2\cdot \ldots\cdot j_{\ell-1}\cdot S_{\ell}}
\newcommand{\gkpattern}{\bar{S}_1\cdot \bar{j}_1\cdot \bar{S}_2\cdot \ldots\cdot \bar{j}_{\ell-1}\cdot \bar{S}_{\ell}}
\newcommand{\gcpattern}{X_1\cdot j_1\cdot X_2\cdot \ldots\cdot j_{\ell-1}\cdot X_{\ell}}
\newcommand{\gckpattern}{\bar{X}_1\cdot \bar{j}_1\cdot \bar{X}_2\cdot \ldots\cdot \bar{j}_{\ell-1}\cdot \bar{X}_{\ell}}

\newcommand{\gqprefix}{S_1\cdot j_1\cdot S_2\cdot \ldots\cdot j_{i-1}\cdot S_i}
\newcommand{\gqrule}{(a_1\cdot (i_2-i_1-1)\cdot\ldots\cdot (i_q-i_{q-1}-1)\cdot a_q)}
\newcommand{\gwspan}{g_{\text{\rm w-span}}}
\newcommand{\gsize}{g_{\text{\rm size}}(\mathcal{P})}
\newcommand{\gmin}{g_{\min}(\mathcal{P})}
\newcommand{\gmax}{g_{\max}(\mathcal{P})}
\newcommand{\gqalgoC}{\textsc{gq-matcher}\xspace}
\newcommand{\gqalgoD}{\textsc{gq-matcher-t}\xspace}
\newcommand{\len}[1]{\text{\rm len}(#1)}
\newcommand{\plen}{\len{\mathcal{P}}}
\newcommand{\klen}[1]{\text{\rm k-len}(#1)}
\newcommand{\pklen}{\klen{\mathcal{P}}}
\newcommand{\weight}{\omega}
\newcommand{\alabel}{\mathit{label}\xspace}
\newcommand{\aroot}{\mathit{root}\xspace}
\newcommand{\ofail}{\mathit{f_o}\xspace}
\newcommand{\T}{\mathcal{T}\xspace}
\newcommand{\fail}{\mathit{fail}\xspace}

\newcommand{\ceil}[1]{\lceil #1 \rceil}
\newcommand{\floor}[1]{\lfloor #1 \rfloor}

\title{Motif matching using gapped patterns\thanks{A preliminary version of this paper appeared in the proceedings of the 24th International Workshop on Combinatorial Algorithm}}

\author{Emanuele Giaquinta\inst{1} \and Kimmo Fredriksson\inst{2} \and Szymon Grabowski\inst{3} \and Alexandru I. Tomescu\inst{1,4} \and Esko Ukkonen\inst{1}}
\institute{Department of Computer Science, University of Helsinki, Finland
	  \email{\{emanuele.giaquinta $\mid$ tomescu $\mid$ ukkonen\}@cs.helsinki.fi} \and
	  School of Computing, University of Eastern Finland
	  \email{kimmo.fredriksson@uef.fi} \and
	  Institute of Applied Computer Science, Lodz University of Technology, Al.\ Politechniki 11, 90--924 {\L}\'od\'z, Poland
	  \email{sgrabow@kis.p.lodz.pl} \and
	  Helsinki Institute for Information Technology HIIT
}

\begin{document}

\sloppy

\maketitle

\begin{abstract}
  We present new algorithms for the problem of \emph{multiple string
    matching} of gapped patterns, where a gapped pattern is a sequence
  of strings such that there is a gap of fixed length between each two
  consecutive strings. The problem has applications in the discovery
  of transcription factor binding sites in DNA sequences when using
  generalized versions of the Position Weight Matrix model to describe
  transcription factor specificities. In these models a motif can be
  matched as a set of gapped patterns with unit-length keywords. The
  existing algorithms for matching a set of gapped patterns are
  worst-case efficient but not practical, or \emph{vice versa}, in
  this particular case. The novel algorithms that we present are based
  on dynamic programming and bit-parallelism, and lie in a
  middle-ground among the existing algorithms. In fact, their time
  complexity is close to the best existing bound and, yet, they are
  also practical. We also provide experimental results which show that
  the presented algorithms are fast in practice, and preferable if all
  the strings in the patterns have unit-length.
\end{abstract}

\section{Introduction}

We consider the problem of matching a set $\mathcal{P}$ of gapped
patterns against a given text of length $n$, where a gapped
pattern is a sequence of strings, over a finite alphabet $\Sigma$
of size $\sigma$, such that there is a gap of fixed length
between each two consecutive strings.
We are interested in computing the list of matching patterns for each
position in the text. This problem is a
specific instance of the \emph{Variable Length Gaps problem}~\cite{DBLP:journals/tcs/BilleGVW12} (VLG problem)
for multiple patterns and has applications in the
discovery of transcription factor (TF) binding sites in DNA sequences when
using generalized versions of the Position Weight Matrix (PWM) model to represent TF binding specificities.
The paper~\cite{GGU} describes how a motif represented as a
generalized PWM can be matched as a set of gapped patterns with unit-length
keywords, and presents algorithms for the restricted case of patterns
with two unit-length keywords.

In the VLG problem a pattern is
a concatenation of strings and of variable-length gaps.
An efficient approach to solve the problem for a single pattern is
based on the simulation of nondeterministic finite
automata~\cite{DBLP:journals/jcb/NavarroR03,DBLP:conf/lata/FredrikssonG09}.
A method to solve the case of one or more patterns is to
translate the patterns into a regular
expression~\cite{DBLP:journals/algorithmica/NavarroR04,DBLP:conf/soda/BilleT10}.
The best time bound for a regular expression is $O(n(k\frac{\log w}{w} +
\log \sigma))$~\cite{DBLP:conf/soda/BilleT10}, where $k$ is the number of
the strings and gaps in the pattern and $w$ is the machine word size in bits.
Observe that in the case of unit-length keywords $k=\Theta(\plen)$,
where $\plen$ is the total number of alphabet symbols in the patterns.
There are also algorithms efficient
in terms of the total number $\alpha$ of occurrences of the strings in
the patterns (keywords) within the
text~\cite{DBLP:journals/jcb/MorgantePVZ05,DBLP:conf/cocoon/RahmanILMS06,DBLP:journals/tcs/BilleGVW12}
\footnote{Note that the number of occurrences of a keyword that occurs in $r$
patterns and in $l$ positions in the text is equal to $r\times l$}.
The best bound obtained for a
single pattern is $O(n\log \sigma + \alpha)$~\cite{DBLP:journals/tcs/BilleGVW12}.
This method can also be extended to multiple patterns. However, if all
the keywords have unit length this result is not ideal, because in
this case $\alpha$ is $\Omega(n\frac{\plen}{\sigma})$ on average if we
assume that the symbols in the patterns are sampled from $\Sigma$
according to a uniform distribution.
A similar approach for multiple patterns~\cite{DBLP:conf/wea/HaapasaloSSS11} leads to
$O(n(\log \sigma + K) + \alpha')$ time, where $K$
is the maximum number of suffixes of a keyword that are also keywords
and $\alpha'$ is the number of text occurrences of pattern
prefixes that end with a keyword. This result may be preferable in
general when $\alpha' < \alpha$. In the case of unit-length keywords, however, a lower bound similar to the one
on $\alpha$ holds also for $\alpha'$, as the prefixes of
unit length have on average $\Omega(n\frac{|\mathcal{P}|}{\sigma})$ occurrences in the text.
Recently, a variant of this algorithm based on word-level parallelism
was presented in~\cite{DBLP:conf/lata/SippuS13}. This algorithm works
in time $O(n(\log \sigma + (\log |\mathcal{P}| +
\frac{k}{w})\alpha_m))$, where $k$ in this case is the maximum number
of keywords in a single pattern and $\alpha_m\ge \ceil{\alpha / n}$ is the maximum number
of occurrences of keywords at a single text position.
When $\alpha$ or $\alpha'$ is large, the bound of~\cite{DBLP:conf/soda/BilleT10} may be preferable.
The drawback of this algorithm is that, to our knowledge, the method
used to implement fixed-length gaps, based on maintaining multiple bit
queues using word-level parallelism, is not practical.

Note that the above bounds do not include preprocessing time and the
$\log \sigma$ term in them is due to the simulation of the Aho-Corasick
automaton for the strings in the patterns.

In this paper we present two new algorithms, based on dynamic
programming and bit-parallelism, for the problem of matching a set of
gapped patterns. The first algorithm has $O(n(\log\sigma + \gwspan
\ceil{\pklen / w}) + occ)$-time complexity, where $\pklen$ is the
total number of keywords in the patterns and $1\le \gwspan\le w$ is
the maximum number of distinct gap lengths that span a single word in
our encoding. This algorithm is preferable only when $\gwspan\ll w$. We then show how to improve
the time bound to $O(n(\log\sigma + \log^2\gsize\ceil{\pklen / w}) +
occ)$, where $\gsize$ is the size of the variation range of the gap
lengths. Note that in the case of unit-length keywords we have $\pklen
= \plen$. This bound is a moderate improvement over the more general
bound for regular expressions by Bille and
Thorup~\cite{DBLP:conf/soda/BilleT10} for $\log\gsize = o(\sqrt{\log
  w})$. This algorithm can also be extended to support character
classes with no overhead. The second algorithm is based on a different
parallelization of the dynamic programming matrix and has
$O(\ceil{n/w}\;\plen+n+occ)$-time complexity. The advantage of this
bound is that it does not depend on the number of distinct gap
lengths. However, it is not strictly on-line, because it processes the
text $w$ characters at a time and it also depends on $\plen$ rather
than on $\pklen$. Moreover, it cannot support character classes without
 overhead. The proposed algorithms obtain a bound similar to the one
of~\cite{DBLP:conf/soda/BilleT10}, in the restricted case of
fixed-length gaps, while being also practical. For this reason, they
provide an effective alternative when $\alpha$ or $\alpha'$ is large.
They are also fast in practice, as shown by experimental evaluation.
A comparison of our algorithms with the existing ones is summarized in Table $1$.

The rest of the paper is organized as follows. In
Section~\ref{sec:Notions} we recall some preliminary notions and
elementary facts. In Section~\ref{sec:motivation} we discuss the
motivation for our work. In Section~\ref{sec:dp} we describe the method
based on dynamic programming for matching a set of gapped patterns and
then in Section~\ref{sec:gq-matcher} and~\ref{sec:gq-matcher-t} we
present the new algorithms based on it. Finally, in
Section~\ref{sec:experiments} we present experimental results to
evaluate the performance of our algorithms.

\begin{table}[t]
\begin{center}
\begin{tabular}{ll}
Time & Reference \\
\hline
$O(n\log\sigma + \alpha)$ & Bille et al.~\cite{DBLP:journals/tcs/BilleGVW12} \\
$O(n(\log\sigma + K) + \alpha')$ & Haapasalo et al.~\cite{DBLP:conf/wea/HaapasaloSSS11} \\
$O(n(\log\sigma + \log w\ceil{\pklen / w}) + occ)$ & Bille and Thorup~\cite{DBLP:conf/soda/BilleT10} \\
$O(n(\log\sigma + \log^2\gsize\ceil{\pklen / w}) + occ)$ & This paper \\
$O(\ceil{n/w}\plen + n + occ)$ & This paper \\
\hline
\end{tabular}
\end{center}

\caption{Comparison of different algorithms for the multiple string
  matching with gapped patterns problem. $\pklen$ and $\plen$ are the
  total number of keywords and symbols in the patterns, respectively.
  $\gsize$ is the size of the variation range of the gap lengths.
  $\alpha \le n\pklen$ and $\alpha' \le n\pklen$ are the total number of
  occurrences in the text of keywords and pattern prefixes,
  respectively. $K \le \pklen$ is the maximum number of suffixes of a
  keyword that are also keywords.}
\end{table}

\section{Basic notions and definitions}\label{sec:Notions}

Let $\Sigma$ denote an integer alphabet of size $\sigma$ and
$\Sigma^*$ the Kleene star of $\Sigma$, i.e., the set of all
possible sequences over $\Sigma$.
$|S|$ is the length of string $S$, $S[i], i \geq 0$, denotes its $(i+1)$-th
character, and $S[i\,\ldots\,j]$ denotes its substring between the $(i+1)$-st and the $(j+1)$-st
characters (inclusive).
For any two strings $S$ and $S'$, we say that $S'$ is a suffix of $S$
(in symbols, $S' \sqsupseteq S$) if $S' = S[i\,\ldots\,|S|-1]$, for some
$0\le i < |S|$.

A gapped pattern $P$ is of the form
$$
\gpattern\,,
$$
where $S_i\in\Sigma^*$, $|S_i|\ge 1$, is the $i$-th string (keyword) and $j_i\ge 0$ is the
length of the gap between keywords $S_i$ and $S_{i+1}$, for $i=1,\ldots,\ell$.
We say that $P$ occurs in a string $T$ at ending position $i$
if
$$
T[i-m+1\,\ldots\,i]=S_1\cdot A_1\cdot S_2\cdot\ldots \cdot A_{\ell-1}\cdot S_{\ell}\,,
$$
where $A_i\in \Sigma^*$, $|A_i| = j_i$, for $1\le i\le \ell - 1$, and $m = \sum_{i=1}^{\ell} |S_i| + \sum_{i=1}^{\ell-1} j_i$.
In this case we write $P\sqsupseteq_g T_i$.
We denote by $\len{P} = \sum_{i=1}^{\ell}|S_i|$ and $\klen{P} = \ell$ the number of alphabet
symbols and keywords in $P$, respectively.
The gapped pattern $P_i = \gqprefix$ is the prefix of $P$ of length $i\le \ell$.
Given a set of gapped patterns $\mathcal{P}$, we denote by
$\plen=\sum_{P\in\mathcal{P}}\len{P}$ and
$\pklen=\sum_{P\in\mathcal{P}}\klen{P}$ the total number of
symbols and keywords in the patterns, respectively.

The RAM model is assumed, with words of size $w$ in bits.
We use some bitwise operations following the standard notation
as in the C language: $\&$, $|$, $\sim$, $\ll$ for \texttt{and}, \texttt{or},
\texttt{not} and \texttt{left shift}, respectively.
The function to compute the position of the most significant non-zero bit of a word $x$
is $\floor{\log_2(x)}$.

Given a set $\mathcal{S}$ of strings over a finite alphabet $\Sigma$,
the \emph{trie} $\T$ associated with $\mathcal{S}$ is a rooted
directed tree, whose edges are labeled by single characters of
$\Sigma$, such that
\begin{enumerate}[(i)]
\item distinct edges out of the same node are labeled
by distinct characters,

\item all paths in $\T$ from the root are
labeled by prefixes of the strings in $\mathcal{S}$,

\item for each string $S$ in $\mathcal{S}$ there exists a path in $\T$
from the root which is labeled by $S$.
\end{enumerate}

Let $Q$ denote the set of nodes of $\T$, $\aroot$ the root of $\T$,
and $\alabel(q)$ the string which labels the path from $\aroot$ to
$q$, for any $q\in Q$. The Aho-Corasick
(AC) automaton~\cite{DBLP:journals/cacm/AhoC75} $(Q, \Sigma, \delta,
\aroot, F)$ for the language $\bigcup_{S\in\mathcal{S}}\Sigma^* S$ is
induced directly by the trie $\T$ for $\mathcal{S}$. The set $F$ of
final states include all the states $q$ such that the set $\{
S\in\mathcal{S}\ |\ S\sqsupseteq \alabel(q)\}$ of strings in
$\mathcal{S}$ which are suffixes of $\alabel(q)$ is nonempty. The
transition function $\delta(q,c)$ of the AC automaton is defined as
the unique state $q'$ such that $\alabel(q')$ is the longest suffix of
$\alabel(q)\cdot c$. Let $\fail(q)$ be the unique state $p$ such that
$\alabel(p)$ is the longest proper suffix of $\alabel(q)$, for any
$q\in Q\setminus\{\aroot\}$. Any transition $\delta(q,c)$ can be
recursively computed as
$$
\delta(q,c) =
\begin{cases}
\delta_{\T}(q,c) & \text{if } \delta_{\T}(q,c) \text{ is defined}\,, \\
\delta(\fail(q), c) & \text{if } q\neq\aroot\,, \\
\aroot & \text{otherwise}\,,
\end{cases}
$$
where $\delta_{\T}$ is the transition function of the trie. Given a
string $T$ of length $n$, let $q_{-1} = \aroot$ and
$q_i=\delta(q_{i-1}, T[i])$ be the state of the AC automaton after
reading the prefix $T[0\,\ldots\,i]$ of $T$, for $0\le i < n$. If the
transitions of the trie are indexed using a balanced binary search
tree, the sequence of states $q_0, \ldots, q_{n-1}$, i..e, the
simulation of the AC automaton on $T$, can be computed in time
$O(n\log\sigma)$.

\section{Motivation}\label{sec:motivation}

Given a DNA sequence and a motif that describes the binding
specificities of a given transcription factor, we study the problem of
finding all the binding sites in the sequence that match the motif.
The traditional model used to represent transcription factor motifs is
the Position Weight Matrix (PWM).
This
model assumes that there is no correlation between positions in the
sites, that is, the contribution of a nucleotide at a given position
to the total affinity does not depend on the other nucleotides which
appear in other positions. The problem of matching the locations in
DNA sequences at which a given transcription factor binds to is well
studied under the PWM model~\cite{DBLP:journals/tcs/PizziU08}. Many
more advanced models have been proposed to overcome the independence
assumption of the PWM (see~\cite{Bi11} for a discussion
on the most important ones). One approach, common to some models,
consists in extending the PWM model by assigning weights to sets of
symbol-position pairs rather than to a single pair only. We focus on
the Feature Motif Model (FMM)~\cite{DBLP:journals/ploscb/SharonLS08}
since, to our knowledge, it is the most general one. In this model the
TF binding specificities are described with so-called \emph{features},
i.e., rules that assign a weight to a set of associations between
symbols and positions. Given a DNA sequence, a set of features and a
motif of length $m$, the matching problem consists in computing the score
of each site (substring) of length $m$ in the sequence, where the
score of a site is the sum of the weights of all the features that
occur in the site. Formally, a \emph{feature} can be denoted as
$$
\{(a_1, i_1), \ldots, (a_q, i_q)\}\rightarrow \weight\,,
$$
where $\weight$ is the affinity contribution of the feature and $a_j\in\{A,
C, G, T\}$ is the nucleotide which must occur at position $i_j$, for
$j=1,\ldots, q$ and $1\le i_j\le m$.
It is easy to transform these rules into new rules
where the left side is a gapped pattern: if $i_1 < i_2 < \ldots
< i_q$, we can induce the following gapped pattern rule
$$
\gqrule\rightarrow (i_q, \weight).
$$
Note that we maintain the last position $i_q$ to recover the original feature.
This transformation has the advantage that the resulting pattern is
position independent. Moreover, after this transformation,
different features may share the same gapped pattern. Hence,
the matching problem can be decomposed into two components: the first
component identifies the occurrences of the groups of features by
searching for the corresponding gapped patterns, while the
second component computes the score for each candidate site using the
information provided by the first component.
For a motif of length $m$, the second component can be easily
implemented by maintaining the score for $m$ site alignments
simultaneously with a circular queue of length $m$. Each time a group
of features with an associated set of position/weight pairs
$\{(i_1,\weight_1),\ldots,(i_r,\weight_r)\}$ is found at position $j$
in the sequence, the algorithm adds the weight $\weight_k$ to the
score of the alignment that ends at position $j+m-i_k$ in the
sequence, if $j\ge i_k$.

\section{Dynamic Programming}\label{sec:dp}

In this section we present a method based on dynamic programming (DP)
to search for a set $\mathcal{P}$ of gapped patterns in a text $T$ of
length $n$. Then, in the next two sections, we show how to
parallelize the computation of the DP matrix column-wise and row-wise
using word-level parallelism.
Let $P$ be a gapped pattern.
We define the matrix $D$ of size $\klen{P}\times n$ where
$$
D_{l,i} =
\begin{cases}
1 & \text{if } P_l\sqsupseteq_{g} T_i\,, \\
0 & \text{otherwise}\,,
\end{cases}
$$
for $0\le l < \klen{P}$ and $0\le i < n$.
For example, the matrix corresponding to $P = c\cdot 2\cdot at\cdot 1\cdot t, T = atcgctcatat$ is
\begin{center}
\begin{tabular}{|l|lllllllllll|}
\hline
      & a & t & c & g & c & t & c & a & t & a & t \\
\hline
c & 0 & 0 & 1 & 0 & 1 & 0 & 1 & 0 & 0 & 0 & 0 \\
at & 0 & 0 & 0 & 0 & 0 & 0 & 0 & 0 & 1 & 0 & 1 \\
t & 0 & 0 & 0 & 0 & 0 & 0 & 0 & 0 & 0 & 0 & 1 \\
\hline
\end{tabular}
\end{center}
From the definition of $D$ it follows that the pattern $P$ occurs
in $T$ at position $i$ if and only if $D_{\klen{P},i} = 1$.
The matrix $D$ can be computed using the recurrence
$$
D_{l,i} = 
\begin{cases}
1	& \text{if } S_l\sqsupseteq T[0\,\ldots\,i-1] \text{ and } (l = 1 \text{ or } D_{l-1,i - |S_l| - j_{l-1}} = 1)\,, \\
0	& \text{otherwise}\,.
\end{cases}
$$
Let $D^k$ be the matrix of the $k$-th pattern in $\mathcal{P}$. This
method can be generalized to multiple patterns by concatenating the
matrices $D^k$ for all the patterns into a single matrix $D$ of size
$\klen{\mathcal{P}}\times n$ and adjusting the definitions
accordingly.
We now sketch the intuition behind the column-wise and row-wise parallelization.

Consider a column-wise computation of $D$. If, for each
$P\in\mathcal{P}$, we replace each gap length $j_i$ in $P$ with
$\bar{j}_i = j_i + |S_{i+1}|$, for $i = 1, \ldots, \klen{P} - 1$, and
let $G$ be the set of distinct gap lengths in $\mathcal{P}$, then we
have that each column of $D$ depends on $|G|$ previous columns.
For example, in the case of $c\cdot 2\cdot at\cdot 1\cdot t$, we have
$\bar{j}_1 = 4, \bar{j}_2 = 2$ and the $l$-th column depends on
columns $l-2$ and $l-4$. Instead, in the case of $c\cdot 2\cdot a\cdot
1\cdot at$ we have $\bar{j}_1 = 3, \bar{j}_2 = 3$ and the $l$-th
column depends on column $l-3$ only. The idea in the column-wise
parallelization is to process $w$ cells of a column in $O(\gwspan)$
time, where $1\le \gwspan\le w$ is the maximum number of distinct gap
lengths that span a segment of $w$ cells in a column. The total time
to compute one column ($n$ in total) is thus $O(\gwspan \ceil{\pklen /
  w})$. We also describe how to obtain an equivalent set of patterns
with $O(\log\gsize)$ distinct gap lengths, where $\gsize = \max G -
\min G + 1$, at the price of $O(\log\gsize)$ new keywords per gap,
thus achieving $O(\log^2\gsize \ceil{\pklen / w})$ time.

Consider now a row-wise computation of $D$. We have that each row of
$D$ depends on the previous row only. To perform this computation
efficiently, we split, for each $P\in\mathcal{P}$, each keyword $S_i$
in $P$ in $|S_i|$ unit-length keywords by inserting a $0$ gap length
between each two consecutive symbols. For example, $c\cdot 2\cdot
at\cdot 1\cdot t$ becomes $c\cdot 2\cdot a\cdot 0\cdot t\cdot 1\cdot
t$ and the corresponding matrix is
\begin{center}
\begin{tabular}{|llllllllllll|}
\hline
      & a & t & c & g & c & t & c & a & t & a & t \\
c & 0 & 0 & 1 & 0 & 1 & 0 & 1 & 0 & 0 & 0 & 0 \\
a & 0 & 0 & 0 & 0 & 0 & 0 & 0 & 1 & 0 & 1 & 0 \\
t & 0 & 0 & 0 & 0 & 0 & 0 & 0 & 0 & 1 & 0 & 1 \\
t & 0 & 0 & 0 & 0 & 0 & 0 & 0 & 0 & 0 & 0 & 1 \\
\hline
\end{tabular}
\end{center}
In this way the number of rows becomes $\plen$. Then, the idea in
the row-wise parallelization is to process $w$ cells of a row in
$O(1)$ time. The total time to compute one row ($\plen$ in total) is
thus $O(\ceil{n/w})$.

\section{Column-wise parallelization}\label{sec:gq-matcher}

Let $P^k$ be the $k$-th pattern in $\mathcal{P}$.
We adopt the superscript notation for $S_i$, $j_i$ and $P_l$ with the same meaning.
We define the set
$$
D_i = \{ (k,l)\ |\ P_l^k\sqsupseteq_{g} T_i \}\,,
$$
of the prefixes of the patterns that occur at position $i$ in $T$, for
$i=0,\ldots,n-1$, $1\le k \le |\mathcal{P}|$ and $1\le l\le
\klen{P^k}$.
The set $D_i$ is a sparse representation of the $i$-th column of the matrix $D$ defined in the previous section.
From the definition of $D_i$ it follows that the pattern $P^k$ occurs
in $T$ at position $i$ if and only if $(k,\klen{P^k})\in D_i$.
For example, if $T = accgtaaacg$ and $\mathcal{P} = \{ cgt\cdot 2\cdot
ac, c\cdot 1\cdot gt\cdot 3\cdot c\}$, we have $D_1 = \{ (2,1) \}$,
$D_4 = \{ (1, 1), (2, 2) \}$ and $D_8 = \{ (1, 2), (2,1), (2,3) \}$
and there is an occurrence of $P^1$ and $P^2$ at position $8$.

Let $\mathcal{K} = \{1,\ldots,\pklen\}$ be the set of indices of the
keywords in $\mathcal{P}$ and let $\bar{T}_i\subseteq \mathcal{K}$ be
the set of indices of the matching keywords in $T$ ending at position
$i$. The sequence $\bar{T}_i$, for $0\le i< n$, is basically a new
text with character classes over $\mathcal{K}$. In the case of the
previous example we have $\mathcal{K} = \{ cgt_1, ac_2, c_3, gt_4, c_5
\}$ and $\bar{T}_1 = \{ ac_2, c_3, c_5\}$, $\bar{T}_4 = \{cgt_1, gt_4\}$
and $\bar{T}_8 = \{ ac_2, c_3, c_5 \}$ (we also show the keyword
corresponding to each index for clarity).

We replace each pattern $\gpattern$ in $\mathcal{P}$ with the pattern
$
\gkpattern\,,
$
with unit-length keywords over the alphabet $\mathcal{K}$, where
$\bar{S}_i\in \mathcal{K}$ and $\bar{j}_i=j_i+|S_{i+1}|$, for $1\le
i< \ell$.
For $\mathcal{P} = \{ cgt\cdot 2\cdot ac, c\cdot 1\cdot gt\cdot 3\cdot
c \}$, the new set is $\{ cgt_1\cdot 4\cdot ac_2, c_3\cdot 3\cdot
gt_4\cdot 4\cdot c_5 \}$.

The sets $D_i$ can be computed using the following lemma:
\begin{lemma}
Let $\mathcal{P}$ and $T$ be a set of gapped patterns and a text of length $n$, respectively.
Then $(k,l) \in D_i$, for $1\le k\le |\mathcal{P}|$, $1\le l\le \klen{P^k}$ and
$i=0,\ldots,n-1$, if and only if
$$
(l=1 \text{ or } (k,l-1)\in D_{i-\bar{j}_{l-1}^k}) \text{ and } \bar{S}_l^k\in \bar{T}_i.
$$
\end{lemma}
The idea is to match the transformed patterns against the text
$\bar{T}$. Let $\gmin$ and $\gmax$ denote the minimum and maximum gap
length in the patterns, respectively. We also denote with $\gsize =
\gmax-\gmin+1$ the size of the variation range of the gap lengths.
We now present how to efficiently compute any column $D_i$ using Lemma $1$ and word-level parallelism.

\begin{figure}[!t]
\begin{center}
\begin{footnotesize}
\begin{minipage}[c]{0.495\textwidth}
\begin{tabular}{|rl|}
  \hline
  \multicolumn{2}{|l|}{\gqalgoC-preprocess ($\mathcal{P}$, $T$)}\\[0.2cm]
  \setcounter{instr}{0}
  \ninstr & $(\delta, \aroot, \bpstyle{B}, \ofail)\leftarrow AC(\mathcal{P})$ \\
  \ninstr & $G\leftarrow \emptyset $ \\
  \ninstr & $m\leftarrow \pklen$ \\
  \ninstr & $\bpstyle{I}\leftarrow 0^m, \bpstyle{M}\leftarrow 0^m$ \\
  \ninstr & \textbf{for} $g=0,\ldots,\gmax$ \textbf{do} $\bpstyle{C}(g)\leftarrow 0^m$ \\
  \ninstr & $l\leftarrow 0$ \\
  \ninstr & \textbf{for} $\gpattern\in \mathcal{P}$ \textbf{do} \\
  \ninstr & \qquad $\bpstyle{I}\leftarrow \bpstyle{I}\OR 1\ll l$ \\
  \ninstr & \qquad \textbf{for} $k=1,\ldots,\ell$ \textbf{do} \\
  \ninstr & \qquad \qquad \textbf{if} $k = \ell$ \textbf{then} \\
  \ninstr & \qquad \qquad \qquad $\bpstyle{M}\leftarrow \bpstyle{M}\OR 1\ll l $ \\
  \ninstr & \qquad \qquad \textbf{else} $g\leftarrow j_{k} + |S_{k+1}|$\\
  \ninstr & \qquad \qquad \qquad $\bpstyle{C}(g)\leftarrow \bpstyle{C}(g)\OR 1\ll l$ \\
  \ninstr & \qquad \qquad \qquad $G\leftarrow G\cup \{ g \}$ \\
  \ninstr & \qquad \qquad $l\leftarrow l + 1$\\
  & \\
  \hline
  \end{tabular}
\end{minipage}
\hfill
\begin{minipage}[c]{0.475\textwidth}
\begin{tabular}{|rl|}
  \hline
  \multicolumn{2}{|l|}{\gqalgoC-search ($\mathcal{P}$, $T$)}\\[0.2cm]
  \setcounter{instr}{0}
  \ninstr & $q\leftarrow \aroot$\\
  \ninstr & \textbf{for} $i=0,\ldots,|T|-1$ \textbf{do} \\
  \ninstr & \qquad $q\leftarrow \delta(q, T[i]), \bpstyle{H}\leftarrow 0^m$ \\
  \ninstr & \qquad \textbf{for} $g\in G$ \textbf{do} \\
  \ninstr & \qquad \qquad $\bpstyle{H}\leftarrow \bpstyle{H}\OR (\bpstyle{D}_{i-g}\AND \bpstyle{C}(g))$ \\
  \ninstr & \qquad $\bpstyle{D}_i\leftarrow ((\bpstyle{H}\ll 1)\OR \bpstyle{I})\AND \bpstyle{B}(\ofail(q))$ \\
  \ninstr & \qquad $\bpstyle{H}\leftarrow \bpstyle{D}_i\AND \bpstyle{M}$ \\
  \ninstr & \qquad {\sc report}($\bpstyle{H}$) \\
  & \\
  \hline
  & \\[0.14cm]
  \multicolumn{2}{|l|}{\textsc{report}($\bpstyle{H}$)}\\[0.2cm]
  \setcounter{instr}{0}
  \ninstr & \textbf{while} $\bpstyle{H}\neq 0^m$ \textbf{do} \\
  \ninstr & \qquad $k\leftarrow \floor{\log_2(\bpstyle{H})}$ \\
  \ninstr & \qquad \textbf{report}($k$) \\
  \ninstr & \qquad $\bpstyle{H}\leftarrow \bpstyle{H}\AND \NOT(1\ll k)$ \\
  \hline
\end{tabular}
\end{minipage}
\end{footnotesize}
\end{center}
\caption{The \gqalgoC algorithm.}
\label{algo:gq3}
\end{figure}

Let $Q$ denote the set of states of the AC automaton for the set of distinct keywords in $\mathcal{P}$.
We store for each state $q$ a pointer $\ofail(q)$ to the state
$q'$ such that $\alabel(q')$ is the longest suffix of $\alabel(q)$
that is also a keyword, if any. Let
$$
B(q)=\{ (k,l)\ |\ S_l^k \sqsupseteq \alabel(q) \}
$$
be the set of all the occurrences of keywords in the patterns in $\mathcal{P}$ that are suffixes of $\alabel(q)$, for any $q\in Q$.
We preprocess $B(q)$ for each state $q$ such that $\alabel(q)$ is a keyword and compute it for any other state using $B(\ofail(q))$.
The sets $B$ can be preprocessed as follows: each time we add to the
AC automaton a keyword with index $(k,l)$ and corresponding state $q$,
we first initialize $B(q)$ to $\emptyset$, if $q$ is created during
the insertion of this keyword, and then add $(k,l)$ to $B(q)$. After
the AC automaton is built, we perform a breadth-first traversal of the
states of the automaton, and for each state $q$ visited such that $\alabel(q)$ is a keyword we set $B(q) =
B(q)\cup B(f_o(q))$.
It is not hard to see that $B(\ofail(q_i))$ encodes the set
$\bar{T}_i$, where $q_i$ is the state of the AC automaton after
reading the prefix $T[0\,\ldots\,i]$ of $T$.

We describe next how to compute any set $D_i$ using word-level parallelism.
Let $G$ be the set of all the distinct gap lengths in the patterns.
In addition to the sets $B(q)$, we preprocess also a set $C(g)$, for each $g\in G$, defined as follows:
$$
C(g) = \{(k,l)\ |\ \bar{j}_{l}^k = g\}\,,
$$
for $1\le k\le |\mathcal{P}|$ and $1\le l< \klen{P^k}$.
For example, for the set $\{ cgt_1\cdot 4\cdot ac_2, c_3\cdot 3\cdot
gt_4\cdot 4\cdot c_3 \}$ we have $C(4) = \{ (1, 1), (2,2)\}$ and $C(3) = \{ (2,1) \}$.
We encode the sets $D_i$, $B(q)$ and $C(g)$ as bit-vectors of $\pklen$ bits.
The generic element $(k,l)$ is mapped onto bit $\sum_{i=1}^{k-1}
\klen{P^i} + \klen{P^k_{l-1}}$, where $\klen{P^k_0} = 0$ for any $k$.
We denote with $\bpstyle{D}_i$,
$\bpstyle{B}(q)$ and $\bpstyle{C}(g)$ the bit-vectors representing
the sets $D_i$, $B(q)$ and $C(g)$, respectively. We also compute two additional
bit-vectors $\bpstyle{I}$ and $\bpstyle{M}$, such that the bit
corresponding to the element $(k,1)$ in $\bpstyle{I}$ and $(k,\klen{P^k})$ in $\bpstyle{M}$ is set
to $1$, for $1\le k\le |\mathcal{P}|$. We basically mark the first and the last bit of each pattern, respectively.
Let $\bpstyle{H}_i$ be the bit-vector equal to the bitwise \texttt{or} of the bit-vectors
\begin{equation}\label{eq:h}
\bpstyle{D}_{i-g}\AND \bpstyle{C}(g)\,,
\end{equation}
for each $g\in G$.
Then the corresponding set $H_i$ is equal to
$$
\bigcup_{g\in G}\{ (k,l)\ |\ (k,l)\in D_{i-g}\wedge \bar{j}_{l}^k = g \}\,.
$$
The bit-vector $D_i$ can then be computed using the following bitwise operations:
$$
\begin{array}{rcl}
    \bpstyle{D}_i &\leftarrow& ((\bpstyle{H}_i\ll 1)\OR \bpstyle{I})\AND \bpstyle{B}(\ofail(q_i)) \\
\end{array}
$$
which correspond to the relation
$$
\{ (k,l)\ |\ ((k,l-1)\in H_i \vee l=1) \wedge (k,l)\in B(\ofail(q_i))\}\,.
$$
To report all the patterns that match at position $i$ it is enough to
iterate over all the bits set in $\bpstyle{D}_i\AND \bpstyle{M}$.
The algorithm, named \gqalgoC, is given in Figure~\ref{algo:gq3}.

The bit-vector $\bpstyle{H}_i$ can be constructed in time
$O(\gwspan\ceil{\pklen / w})$, $1\le \gwspan\le w$, as follows: we
compute Equation~\ref{eq:h} for each word of the bit-vector
separately, starting from the least significant one. For a given word
with index $j$, we have to compute equation~\ref{eq:h} only for each
$g\in G$ such that the $j$-th word of $C(g)$ has at least one bit set.
Each position in the bit-vector is spanned by exactly one gap, so the
number of such $g$ is at most $w$. Hence, if we maintain, for each
index $j$, the list $G_j$ of all the distinct gap lengths that span
the positions of the $j$-th word, we can compute $\bpstyle{H}_i$ in
time $\sum_{j=1}^{\ceil{\pklen / w}}|G_j|$, which yields the
advertised bound by replacing $|G_j|$ with $\gwspan = \max_j |G_j|$.

The bit-vectors $\bpstyle{B}(f_o(q_i))$ encoding the sets $\bar{T}_i$,
for $0\le i < n$, can be computed in $O(n\log\sigma)$ time using the
AC automaton for the set of distinct keywords in $\mathcal{P}$. Given
the bit-vectors $\bpstyle{H}_i$ and $\bpstyle{B}(f_o(q_i))$, the
bit-vector $\bpstyle{D}_i$ can be computed in $O(\ceil{\pklen / w})$
time. The time complexity of the searching phase of the algorithm is
then $O(n(\log\sigma + \gwspan \ceil{\pklen / w}) + occ)$.

The AC automaton requires $\Theta(\plen)$ space.
Moreover, for the recursion of Lemma $1$, the algorithm needs to keep
the sets $D$ computed in the last $\gmax$ iterations. The lists $G_j$
require $O(\pklen + \ceil{\pklen / w})$ space in total. Finally, the
number of $B$ sets (which corresponds to the number of distinct
keywords) is $\le \pklen$ while the number of $C$ sets is $\le \gmax$.
Hence, the space complexity is $O(\plen + (\gmax + \pklen)\ceil{\pklen
  / w})$.

Observe that the size of the sets $G_j$ depends also on the ordering
of the patterns (unless $\klen{P}$ is a multiple of $w$ for each
$P\in\mathcal{P}$), since more than one pattern can be packed into the
same word. Hence, it can be possibly reduced by finding an ordering
that maps onto the same word patterns that share many gap lengths. We
now show that the problem of minimizing $\sum_j |G_j|$ is hard. In
order to formally define the problem, we introduce the following
definition:

\newcommand{\tablereduction}{
\renewcommand{\tabcolsep}{0.4mm}
\small
\begin{tabular}{rrrrr|rrrr|rrrr|rrrr}
$L_1$ = & \textbf{1} & \textbf{2} & 5 & 6 & \textbf{1} & \textbf{2} & 5 & 6 & \textbf{1} & \textbf{2} & 5 & 6 & \textbf{1} & \textbf{2} & 5 & 6\\
$L_2$ = & \textbf{1} & \textbf{3} & 7 & 8 & \textbf{1} & \textbf{3} & 7 & 8 & \textbf{1} & \textbf{3} & 7 & 8 & \textbf{1} & \textbf{3} & 7 & 8\\
$L_3$ = & \textbf{2} & \textbf{3} & \textbf{4} & 9 & \textbf{2} & \textbf{3} & \textbf{4} & 9 & \textbf{2} & \textbf{3} & \textbf{4} & 9 & \textbf{2} & \textbf{3} & \textbf{4} & 9\\
$L_4$ = & \ \textbf{4} & 10 & 11 & 12 & \ \textbf{4} & 10 & 11 & 12 & \ \textbf{4} & 10 & 11 & 12 & \ \textbf{4} & 10 & 11 & 12
\end{tabular}
}

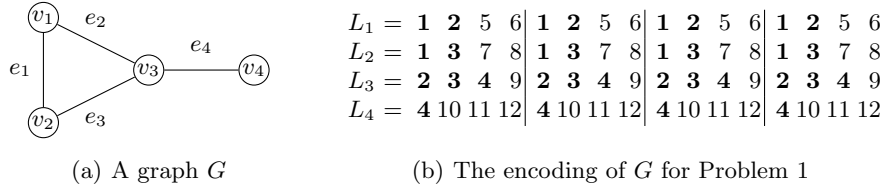
\begin{figure}
\subfigure[A graph $G$]{
\begin{tikzpicture}[scale=.7,-]
\draw[clip,draw=none] (-1,-1.5) rectangle (5,1.5);
\node[circle,draw=black,inner sep = -2mm,minimum size = 4mm,fill=white] (v1) at (0,1) {$v_1$};
\node[circle,draw=black,inner sep = -2mm,minimum size = 4mm,fill=white] (v2) at (0,-1) {$v_2$};
\node[circle,draw=black,inner sep = -2mm,minimum size = 4mm,fill=white] (v3) at (2,0) {$v_3$};
\node[circle,draw=black,inner sep = -2mm,minimum size = 4mm,fill=white] (v4) at (4,0) {$v_4$};

\draw (v1) to node[draw=none,fill=none,left] {$e_1$} (v2);
\draw (v1) to node[draw=none,fill=none,above] {$e_2$} (v3);
\draw (v2) to node[draw=none,fill=none,below] {$e_3$} (v3);
\draw (v3) to node[draw=none,fill=none,above] {$e_4$} (v4);
\end{tikzpicture}
}
\subfigure[The encoding of $G$ for Problem~\ref{problem-nphard}]{
{\raisebox{10mm}\tablereduction}
}
\caption{The reduction of the Hamiltonian Path Problem to Problem~\ref{problem-nphard}. The encoding of the graph $G$ with $n = 4$ vertices and $m = 4$ edges has $U = \{1,\dots,m\} \cup \{m+1,\dots,n^2-m\}$, and for every vertex $v_i$, there is a list $L_i$ made up of $n$ copies of a sublist of length $m$ consisting of the indices of its incident edges plus some unique symbols from $\{m+1,\dots,n^2-m\}$; we take $b = (n+1)m$ and $M = (2m - 1)(n-1) + m$.
\label{fig-reduction}}
\end{figure}

\begin{definition}
Let $L_1, L_2, \ldots, L_n$ be a sequence of lists of integers and let $L_c$ be
the list resulting from their concatenation, say $L_c = l_1,\dots,l_{|L_c|}$. For a given integer $b$, we define the $b$-mapping of the lists as the sequence of lists
$L_1^{b},L_2^{b},\ldots,L_r^{b}$ where $r=\ceil{|L_c|/b}$, list
$L^{b}_i$ contains the elements $l_{(i-1)b + 1},l_{(i-1)b + 2},\ldots,l_{(i-1)b + b}$ of $L_c$, for $1\le i \leq \floor{|L_c|/b}$, and, if $r > \floor{|L_c|/b}$, list $L^{b}_r$ contains the elements $l_{(r-1)b + 1},l_{(r-1)b + 2},\ldots,l_{(r-1)b + (|L_c|\bmod b)}$.
\end{definition}

Then, the problem of minimizing $\sum_j |G_j|$ can be stated as (where in our case we have $n=|\mathcal{P}|$, $b=w$, $U=G$ and $L_k = j^k_1,j^k_2,\ldots,j^k_{\klen{P^k}}$, for $1\le k\le |\mathcal{P}|$):

\begin{problem}[Permutation with Minimum Distinct Binned Symbols, PMDBS]
Given a sequence of $n$ lists of integers $L_1, L_2, \ldots, L_n$ over a
universe $U$, and an integer $b$, find the permutation $\pi$ of
$1,\ldots,n$ which minimizes the sum, over all lists $L^b$ in the $b$-mapping of $L_{\pi(1)},\dots,L_{\pi(n)}$, of the number of distinct elements in $L^b$.
\label{problem-nphard}
\end{problem}

We claim that problem PMDBS is intractable (the full proof is in the Appendix):

\begin{theorem}
Problem PMDBS is NP-hard in the strong sense.
\end{theorem}
\begin{oproof}
We reduce from the Hamiltonian Path Problem (see \cite{Garey-Johnson} for basic notions and definitions). In the decision version of the Problem~PMDBS, we ask for a permutation $\pi$ of $1,\ldots,n$ such that the sum, over all lists $L^b$ in the $b$-mapping of $L_{\pi(1)},\dots,L_{\pi(n)}$, of the number of distinct elements in $L^b$ is at most a given number $M$. 

The idea behind our reduction is that, given a graph $G$ with $n$
vertices, the vertices of $G$ will be encoded by lists, where the list
of a vertex consists of the indices of the edges incident to it, under
a suitable encoding (see Fig.~\ref{fig-reduction} for an example).
This encoding will be such that, choosing $M$ suitably, a permutation
of $1,\dots,n$ satisfying the bound $M$ corresponds to a Hamiltonian
Path in $G$ and \emph{vice versa}. \qed
\end{oproof}

We now show how to improve the time complexity in the worst-case by
constructing an equivalent set of patterns with $O(\log\gsize)$
distinct gap lengths.
Given a set $S\subset\mathbb{N}$, a set $X\subset\mathbb{N}$ is a
$\gamma$-generating set of $S$ if every element of $S$ can be
expressed as the sum of at most $\gamma$, non necessarily distinct,
elements of $X$. Suppose that $X$ is a $\gamma$-generating set of $G$.
We augment the alphabet $\Sigma$ with a wildcard symbol $*$ that
matches any symbol of the original alphabet and define the function
$$
\phi(g) = (i_1 - 1)\cdot *\cdot (i_2 - 1)\cdot *\cdot \ldots \cdot(i_{l-1}-1)\cdot*\cdot i_l\,,
$$
for $g\in G$, where $\{ i_1, i_2, \ldots, i_l\}$ is an arbitrary
combination with repetitions from $X$ of size $l\le \gamma$ which
generate $g$, i.e., $\sum_{j=1}^{l} i_j = g$. The function $\phi$ maps
a gap length $g$ onto a concatenation of $l$ gap lengths from the set
$X\cup\{ i-1\ |\ i\in X\}$ and $l-1$ wildcard symbols. For example, if
$G = \{ 1, 2, 5, 6, 10\}$ then $X = \{1, 5\}$ is a $2$-generating set of $G$ and
$$
\begin{array}{l}
\phi(1) = 1 \\
\phi(2) = \phi(1 + 1) = 0\cdot *\cdot 1 \\
\phi(5) = 5 \\
\phi(6) = \phi(1 + 5) = 0\cdot *\cdot 5 \\
\phi(10) = \phi(5 + 5) = 4\cdot *\cdot 5 \\
\end{array}
$$
We generate a new set of
patterns $\mathcal{P}'$ from $\mathcal{P}$, by transforming each
pattern $\gkpattern$ in $\mathcal{P}$ into the equivalent pattern
$$
\bar{S}_1\cdot \phi(\bar{j}_1)\cdot \bar{S}_2\cdot \ldots\cdot \phi(\bar{j}_{\ell-1})\cdot \bar{S}_{\ell}\,.
$$
In the next subsection we describe how to extend the algorithm
presented above to support character classes and therefore also
wildcard symbols, since a wildcard is equivalent to a character class
containing all the symbols in $\Sigma$. By definition of $\phi$ we
have that $\klen{\mathcal{P}'} < \gamma\pklen$, since the number of
gaps that are split is at most $\pklen - |\mathcal{P}|$ and the number
of wildcard symbols that are added per gap is at most $\gamma - 1$. The
number of words needed for a bit-vector is then $< \ceil{\gamma\pklen
  / w}\le \gamma\ceil{\pklen / w}$. Moreover, the set $G'$ of distinct
gap lengths in $\mathcal{P}'$ is contained in $X\cup\{i-1\ |\ i\in X\}$ and so its
cardinality is $O(|X|)$.
This construction thus yields a
$O(n(\log\sigma + |X|\gamma\ceil{\pklen / w}) + occ)$ bound, which depends on the generating set used.

W.l.o.g. we assume that $\gmax$ is a power of two (if it is not, we
round it up to the nearest power of two). Any positive integer $g\le
\gmax$ can be expressed as a sum of distinct positive powers of two,
i.e., the binary encoding of $g$, such that the largest power of two
is $\le 2^{\log\gmax}$. This implies that the set $X = \{0\}\cup\{ 2^i\ |\ 0\le
i\le \log\gmax \}$ is a $(\log\gmax + 1)$-generating set of
$G$ (we include $0$ in $X$ because $G$ may contain $0$). For
example, if $G = \{ 1, 2, 5, 6, 10\}$ then $X = \{ 2^i\ |\ 0\le i\le
3\}$ and
$$
\begin{array}{l}
\phi(1) = 1 \\
\phi(2) = 2 \\
\phi(5) = \phi(2^0 + 2^2) = 0\cdot *\cdot 4 \\
\phi(6) = \phi(2^1 + 2^2) = 1\cdot *\cdot 4 \\
\phi(10) = \phi(2^1 + 2^3) = 1\cdot *\cdot 8 \\
\end{array}
$$
This generating set yields a $\log^2\gmax$ factor in the bound, since
$|X|=\log\gmax+2$ and $\gamma=\log\gmax+1$. We now show how to
further improve the bound. Any integer $\gmin\le g \le \gmax$ can be
written as $\gmin + g'$, where $0\le g'\le \gsize$. Hence, based on
the reasoning above, the set $\{\gmin\}\cup\{ 2^i\ |\ 0\le i\le
\log\gsize \}$ is a $(\log\gsize + 2)$-generating set of $G$. We thus
obtain the following result:

\begin{theorem}
  Given a set $\mathcal{P}$ of gapped patterns and a text $T$ of
  length $n$, all the occurrences in $T$ of the patterns in
  $\mathcal{P}$ can be reported in time $O(n(\log\sigma +
  \log^2\gsize\ceil{\pklen / w}) + occ)$.
\end{theorem}

\subsection{Character classes}

In this subsection we describe how to extend the $\gqalgoC$ algorithm
to support character classes in the patterns. Let
$$
\gcpattern
$$
be a gapped pattern with character classes, where the keyword $X_i$ is
either a string or a character class, i.e., a subset of $\Sigma$. We
again replace each pattern $\gcpattern$ with the pattern $\gckpattern$
with unit-length keywords over the alphabet $\{1,\ldots,\pklen\}$,
where $\bar{j}_i = j_i$ if $X_i$ is a character class. Let
$\mathcal{S}_i$ be the set including $X_i$ itself if $X_i$ is a string
and all the symbols in $X_i$ otherwise. A keyword $X_i$ matches in $T$
at ending position $i$, i.e., $\bar{X}_i\in\bar{T}_i$, if there is a
string $S\in\mathcal{S}_i$ such that $S\sqsupseteq T[0\,\ldots\,i]$.
Observe that Lemma $1$ can be used as it is. We build the AC automaton
for the set $\bigcup \mathcal{S}^k_l$, for $1\le k\le |\mathcal{P}|$
and $1\le l\le \klen{P^k}$. To support this generalized pattern it is
enough to change the definition of the sets $B(q)$ as follows:
$$
B(q)=\{ (k,l)\ |\ \exists\, S\in\mathcal{S}_l^k: S \sqsupseteq \alabel(q) \}\,.
$$
Note that all the strings in a given set $\mathcal{S}_l^k$ are mapped
onto the same index $(k,l)$. The algorithm (including the computation
of the sets $B(q)$) does not require any change. Since we add $\sigma$
distinct strings at most in total for the character classes, the
number of $B$ sets is $\le \pklen + \sigma$ and thus we have an
$O(\sigma\ceil{\pklen / w})$ overhead in the preprocessing time and
space complexity.

\section{Row-wise parallelization} \label{sec:gq-matcher-t}

We now describe the row-wise parallelization of the DP matrix, based on the ideas of the
$(\delta,\alpha)$-matching algorithm described in
\cite{DBLP:conf/wea/FredrikssonG06}. This algorithm works for a single
pattern only, thus to solve the multi-pattern case we need to
run (the search phase of) the algorithm several times. In this
algorithm we take a different approach to handle arbitrary length
keywords. In particular, we first transform each pattern $\gpattern$ in $\mathcal{P}$ into the equivalent pattern
$\psi(S_1)\cdot j_1\cdot \psi(S_2)\cdot \ldots\cdot j_{\ell-1}\cdot \psi(S_\ell)$, where
$$
\psi(S) =
\begin{cases}
S[0]\cdot 0\cdot \psi(S[1\,\ldots\,|S|-1]) & \text{ if } |S| > 1\,, \\
S[0] & \text{otherwise}\,, \\
\end{cases}
$$
so that all the keywords have unit length and the number of keywords is $\plen$.
We denote by $p^k_r$ the $r$-th keyword (symbol) of the $k$-th pattern.
We also parallelize over the text, rather than over the set of patterns.
The main benefit is that now there is only one gap length to consider
at each step. This also means that instead of preprocessing the set of
patterns, we now must preprocess the text.
For the same reason the algorithm is not strictly on-line anymore, as it processes 
the text $w$ characters at a time.

Let $D^k$ be the matrix as defined in Section~\ref{sec:dp} for the $k$-th pattern in $\mathcal{P}$ and
let $D^k_{r,c}$ be the cell of $D^k$ at row $r$ and column $c$.
Observe that in the case of unit-length keywords the recurrence to compute $D^k$ simplifies to
$$
D^k_{r,c} = 
\begin{cases}
1	& \text{if } p^k_r = T[c] \text{ and } (r = 1 \text{ or } D^k_{r-1,c-j^k_{r-1}-1} = 1)\,, \\
0	& \text{otherwise}. 
\end{cases}
$$
The matrix $D^k$ has $\len{P^k}$ rows and $n$ columns and is easy to
compute in $O(n\;\len{P^k})$ time using dynamic programming. We now
show how it can be computed in $O(\ceil{n/w}\;\len{P^k})$ time using
word-level parallelism by processing chunks of $w$ columns in $O(1)$
time.

\setcounter{instr}{0}
\begin{figure}[!t]
\begin{center}
\begin{footnotesize}
\begin{tabular}{|rl|}
  \hline
  \multicolumn{2}{|l|}{\gqalgoD ($\mathcal{P}$, $T$)}\\[0.2cm]
  \ninstr & \textbf{for} $s \in \Sigma$ \textbf{do} $\bpstyle{V}[s] \leftarrow 0$ \\
  \ninstr & \textbf{for} $c \leftarrow 0$ \textbf{to} $\ceil{n/w}$ \textbf{do} \\
  \ninstr & \qquad \textbf{for} $i \leftarrow cw$ \textbf{to} $\min(n, (c+1)w)-1$ \textbf{do} $\bpstyle{V}[T[i]] \leftarrow \bpstyle{V}[T[i]] \OR (1 \ll (i \MOD w))$\\
  \ninstr & \qquad \textbf{for} $k \leftarrow 1$ \textbf{to} $|\mathcal{P}|$ \textbf{do} \\
  \ninstr & \qquad\qquad $\bpstyle{D}^{k,w}_{1,c} \leftarrow \bpstyle{V}[p^k_1]$\\
  \ninstr & \qquad\qquad \textbf{for} $r \leftarrow 2$ \textbf{to} $\len{P^k}$ \textbf{do} $\bpstyle{D}^{k,w}_{r,c} \leftarrow \bpstyle{V}[p^k_r] \AND M(k, r-1, c, j_{r-1}+1)$\\
  \ninstr & \qquad\qquad \textsc{report}($\bpstyle{D}^{k,w}_{\len{P^k},c}$)\\
  \ninstr & \qquad \textbf{for} $i \leftarrow cw$ \textbf{to} $\min(n, (c+1)w)-1$ \textbf{do} $\bpstyle{V}[T[i]] \leftarrow 0$\\
  \hline
\end{tabular}
\end{footnotesize}
\end{center}
\caption{The \gqalgoD algorithm.}
\label{algo:gqwea}
\end{figure}

To this end, let $V$ be a matrix of size $\sigma\times n$, where
$$
V_{s,c} =
\begin{cases}
1	& \text{if } s = T[c]\,, \\
0	& \text{otherwise}\,,
\end{cases}
$$
for $s\in \Sigma$ and $0 \le c < n$. Let also $\Sigma_{\mathcal{P}}$ be the subset of $\Sigma$ of size
$\sigma_{\mathcal{P}}\le \min(\sigma,\plen)$ of the symbols occurring
in the patterns.
Assume that we have the rows of $V$ which correspond to the symbols of
$\Sigma_{\mathcal{P}}$ encoded in an array of $\sigma$ bit-vectors of
$\ceil{n/w}$ bits. The entries corresponding to symbols not in
$\Sigma_{\mathcal{P}}$ are not initialized. The set
$\Sigma_{\mathcal{P}}$ can be trivially computed in
$O(\plen\log\sigma_{\mathcal{P}})$ time using a binary search tree.
The array can be computed in $O(\ceil{n/w}\;\sigma_{\mathcal{P}}+n)$
time.

The computation of $D^k$ will proceed row-wise, $w$ columns at once, as
each matrix element takes only one bit of storage and we can store $w$ columns 
into a single machine word. 
We adopt the notation $\bpstyle{D}^{k,w}_{r,c} = D^k_{r,cw \ldots (c+1)w-1}$, and 
analogously for $\bpstyle{V}$.
First notice that by definition $\bpstyle{D}^{k,w}_{1,c} = \bpstyle{V}^{w}_{p^k_1,c}$.
Assume now that the words $\bpstyle{D}^{k,w}_{r-1,c'}$ for $c' \le c$ have been already 
computed, and we want to compute $\bpstyle{D}^{k,w}_{r,c}$.
To do so, we need to check if any text
character in the current chunk $T[cw \ldots (c+1)w-1]$ matches the pattern character $p^k_r$ 
(readily solved as $\bpstyle{V}^w_{p^k_r,c}$), and if $g = j_{r-1}+1$ text characters back there was a matching 
pattern prefix of length $r-1$. The corresponding bits signaling these 
prefix matches, relevant to the current 
chunk, are distributed in at most two consecutive words in a $w$-bit wide interval in 
the previous row, namely in words 
$\bpstyle{D}^{k,w}_{r-1,c'-1}$ and $\bpstyle{D}^{k,w}_{r-1,c'}$, where $c'= c-\floor{g / w}$.
We select the relevant bits and combine them into a single word using the following function:
\[
M(k, r, c, g) = (\bpstyle{D}^{k,w}_{r,c - \floor{g / w} - 1} \gg (w - (g \MOD w))) \OR (\bpstyle{D}^{k,w}_{r,c-\floor{g/w}} \ll (g \MOD w)).
\]
The recurrence can now be written as
\[
\bpstyle{D}^{k,w}_{r,c} \leftarrow \bpstyle{V}^{w}_{p^k_r,c} \AND M(k, r-1, c, j_{r-1}+1),
\]
and $D^k$ can be computed in $O(\ceil{n/w}\;\len{P^k})$ time for any $k$. To check the occurrences,
we just scan the last row of the matrix and report every position where the bit is $1$.
To handle all the patterns, we 
run the search algorithm $|\mathcal P|$ times, which gives $O(\ceil{n/w}\;\plen+n+occ)$ total time,
including the preprocessing. The algorithm needs $O(\sigma +\ceil{\gmax/w} \max_k(\len{P^k}))$ 
words of space, 
as only the current column of $\bpstyle{V}^w$ and the last $O(\ceil{\gmax/w})$ columns of 
$\bpstyle{D}^{k,w}$ need to be kept in memory at any given time.

Based on the observation that we need only the rows of $V$
corresponding to the symbols in $\Sigma_{\mathcal{P}}$, we can also
manage to reduce the space for $V$ from $O(\sigma)$ to
$O(\min(\sigma_{\mathcal{P}},w) + \ceil{\sigma / w})$ words. First, we
build a (constant time) mapping $\mu$ from $\Sigma_{\mathcal{P}}$ to
$\{1,\ldots,\sigma_{\mathcal{P}}\}$. One (practical) way to compute
$\mu$ is to encode $\Sigma_{\mathcal{P}}$ in a bit-vector $\bpstyle{S}$
of $\sigma$ bits and build a rank
dictionary~\cite{DBLP:conf/fsttcs/Munro96} for it. The rank dictionary
allows one to compute the function $rank_1(\bpstyle{S}, i)$ which
returns the number of bits set to $1$ among the first $i$ positions in
$\bpstyle{S}$. In this way the mapping can be implemented as $\mu(s) =
rank_1(\bpstyle{S}, s)$. The rank dictionary can be built in
$O(\sigma)$ time and requires $O(\ceil{\sigma / w})$ space. We can then
encode $V$ using $O(\sigma_{\mathcal{P}})$ words and access the row
corresponding to any symbol $s\in\Sigma_{\mathcal{P}}$ as $V[\mu(s)]$.
If $w < \sigma_{\mathcal{P}}$ we can further reduce the space for $V$
by exploiting the fact that we process $T$ in chunks. The idea is to
compute, for a given chunk of $T$ of length $w$ starting at position
$c$, a bit-vector $\bpstyle{S}'$ of $\sigma_{\mathcal{P}}$ bits where
we set bit $\mu(s)$ for each $s\in \Sigma_{\mathcal{P}}$ which occurs in
the chunk. Note that if $s$ does not occur in the chunk then
$\bpstyle{V}^w_{s,c} = 0$. By building a rank dictionary for
$\bpstyle{S}'$ we obtain a mapping from the subset of
$\Sigma_{\mathcal{P}}$ encoded in $\bpstyle{S}'$ to $\{1,\ldots,w\}$,
i.e., $rank_1(\bpstyle{S}', \mu(s))$ is the mapping for symbol $s$. We
can then encode $V$ using $O(w)$ words and access the row
corresponding to any symbol $s\in\Sigma_{\mathcal{P}}$ as
$V[rank_1(\bpstyle{S}', \mu(s))]$, if bit $\mu(s)$ is set in
$\bpstyle{S}'$, and as a word equal to $0$ otherwise. Observe that
there are $\ceil{n/w}$ chunks; the time to compute any bit-vector
$\bpstyle{S'}$ and its rank dictionary is $O(w +
\sigma_{\mathcal{P}})$. Hence, we spend $O(\ceil{n/w}
\sigma_{\mathcal{P}} + n)$ time in total and maintain the original
time complexity. Alternatively, we can reduce the space for $V$ to
$O(\sigma_{\mathcal{P}})$ by computing $\mu$ using Ru\v{z}i\'c's
dictionary~\cite{DBLP:conf/icalp/Ruzic08} for $\Sigma_{\mathcal{P}}$,
whose construction requires
$O(\sigma_{\mathcal{P}}(\log\log\sigma_{\mathcal{P}})^2)$ time.

The algorithm, named \gqalgoD, is given in Figure~\ref{algo:gqwea}. We
thus obtain the following result:
\begin{theorem}
  Given a set $\mathcal{P}$ of gapped patterns and a text $T$ of
  length $n$, given in chunks of $w$ characters, all the occurrences
  in $T$ of the patterns in $\mathcal{P}$ can be reported in time
  $O(\ceil{n/w}\;\plen + n + occ)$.
\end{theorem}

\begin{figure}
\begin{center}
\subfigure{
\includegraphics[width=0.47\textwidth]{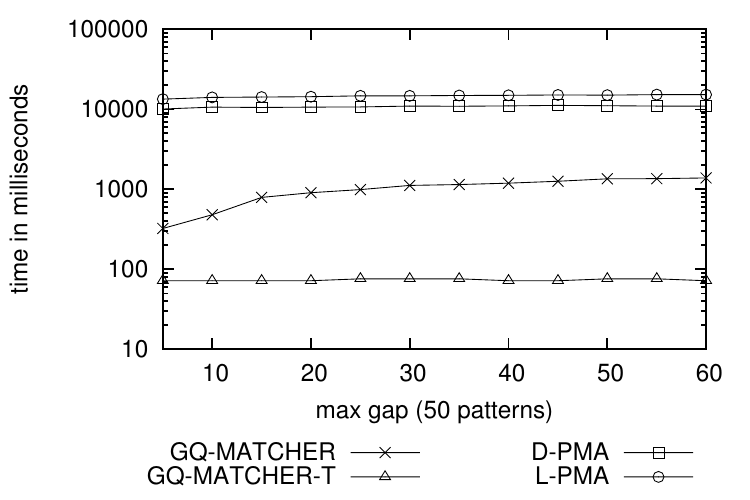}
}
\subfigure{
\includegraphics[width=0.47\textwidth]{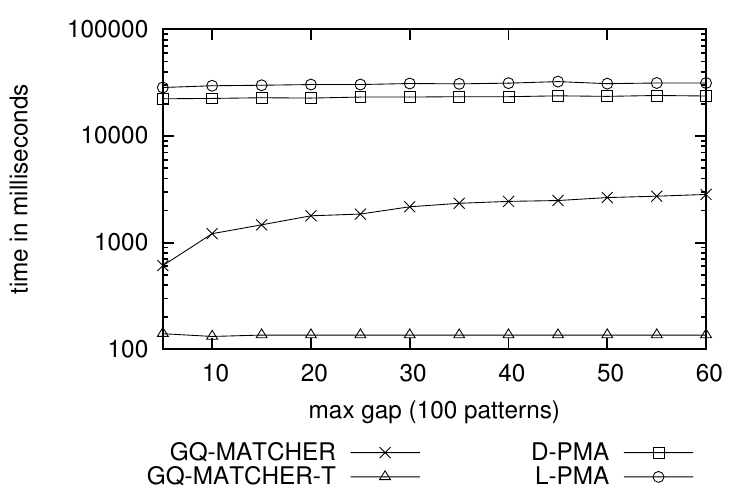}
}
\subfigure{
\includegraphics[width=0.47\textwidth]{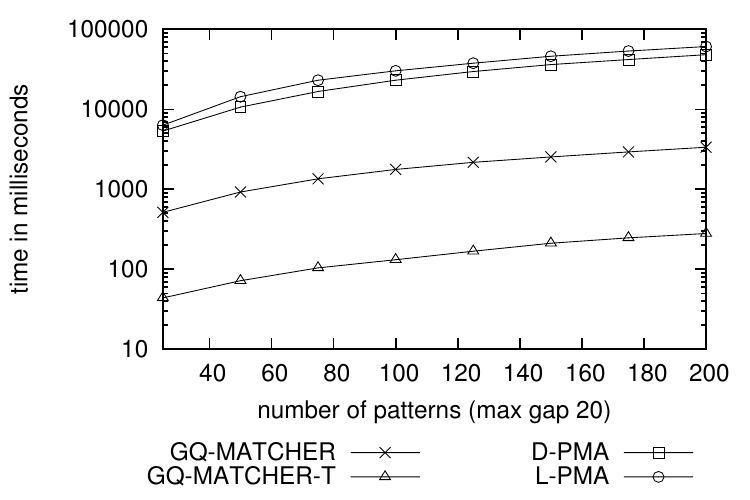}
}
\subfigure{
\includegraphics[width=0.47\textwidth]{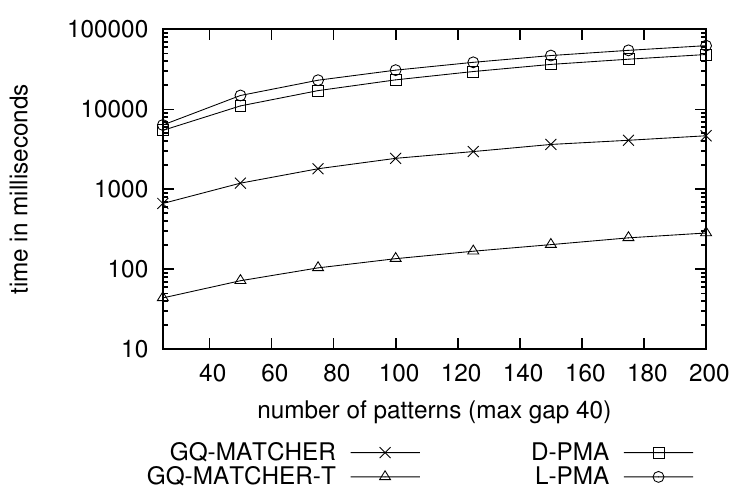}
}
\subfigure{
\includegraphics[width=0.47\textwidth]{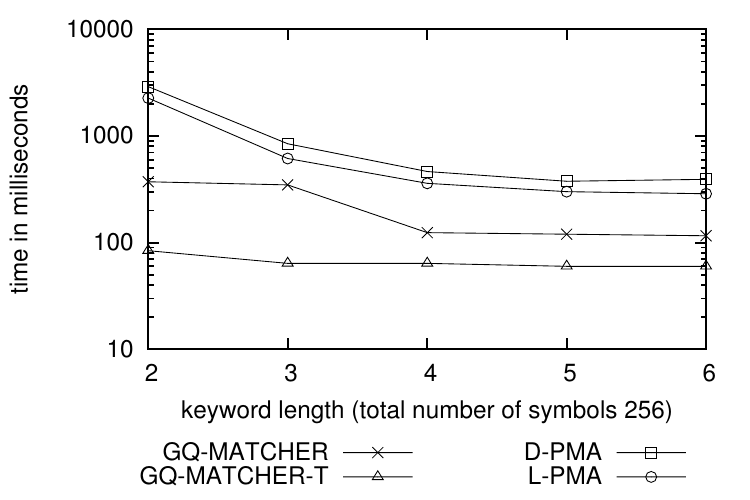}
}
\subfigure{
\includegraphics[width=0.47\textwidth]{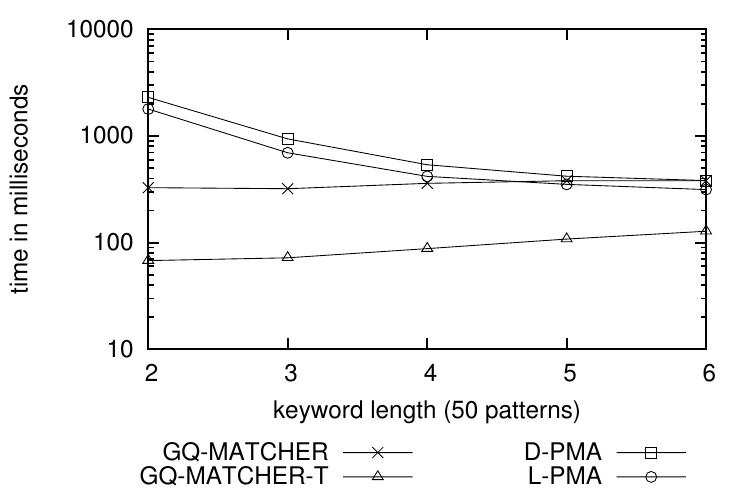}
}
\caption{Experimental results on the DNA sequence of
  the \textit{Escherichia coli} genome with randomly generated gapped
  patterns.
  Top row: $6$ unit-length keywords, varying gap interval with a set of 50 and 100 patterns;
  Middle row: $6$ unit-length keywords, varying number of patterns with maximum gap 20 and 40;
  Bottom row: $2$ keywords, varying keyword length.}
\label{fig:benchq1}
\label{fig:benchq2}
\label{fig:benchq3}
\end{center}
\end{figure}

\begin{figure}
\begin{center}
\subfigure{
\includegraphics[width=0.47\textwidth]{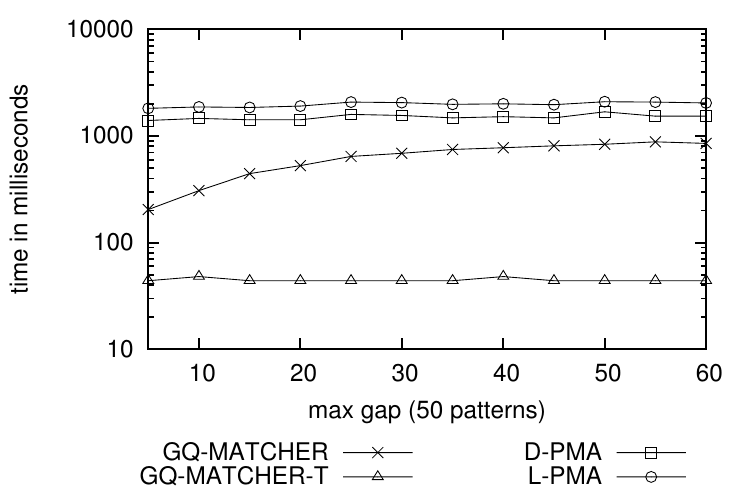}
}
\subfigure{
\includegraphics[width=0.47\textwidth]{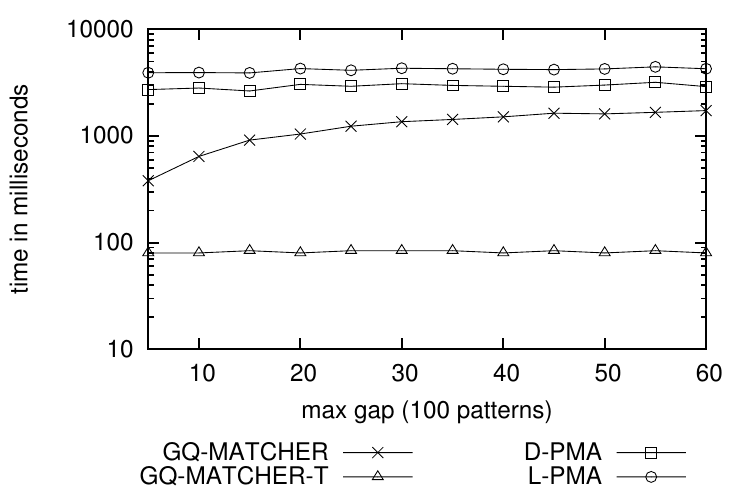}
}
\subfigure{
\includegraphics[width=0.47\textwidth]{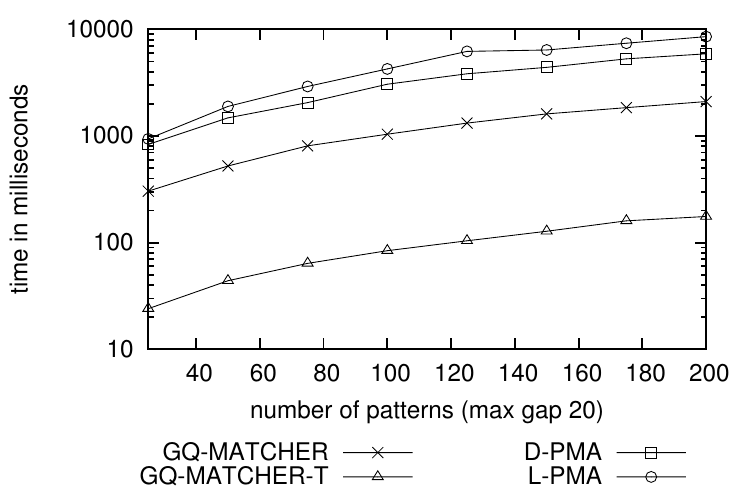}
}
\subfigure{
\includegraphics[width=0.47\textwidth]{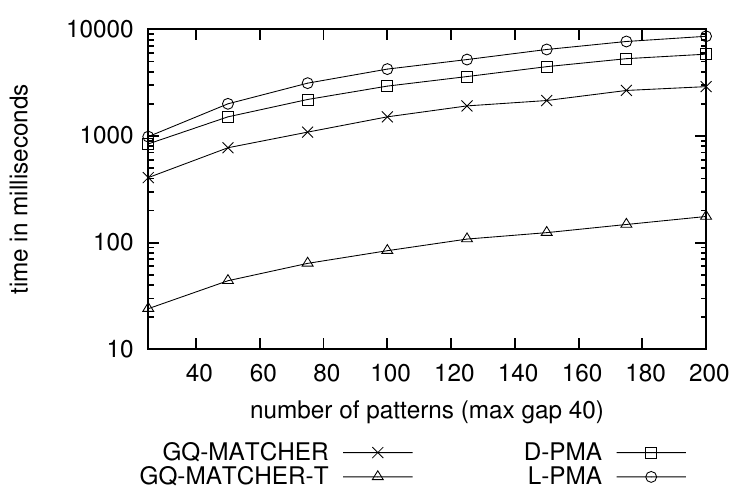}
}
\subfigure{
\includegraphics[width=0.47\textwidth]{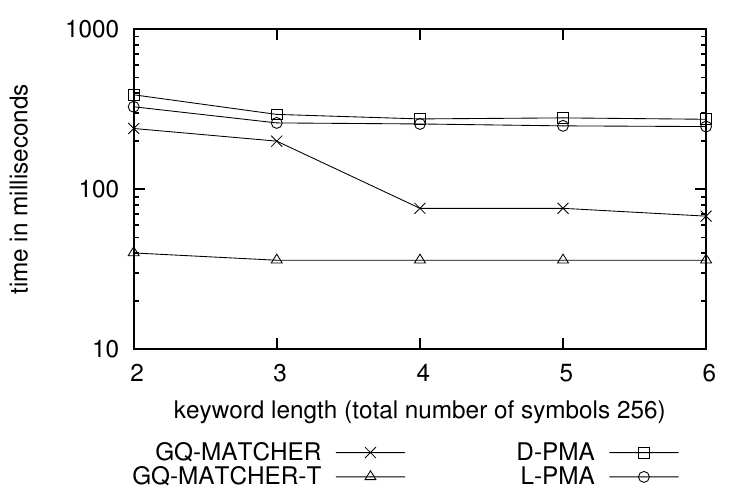}
}
\subfigure{
\includegraphics[width=0.47\textwidth]{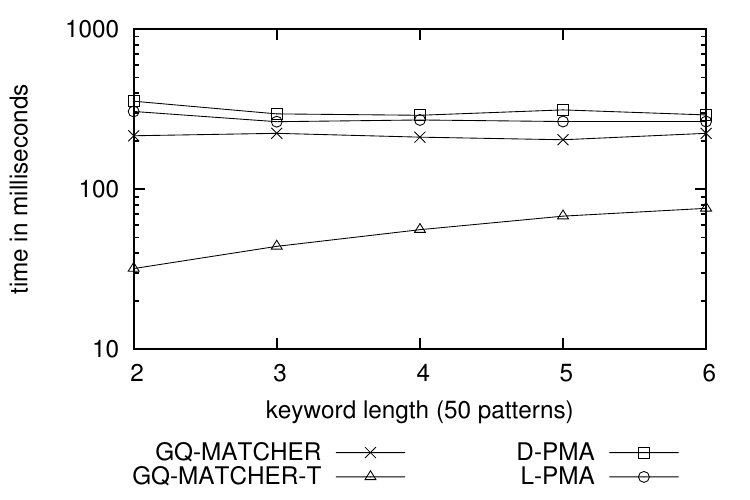}
}
\caption{Experimental results on the protein sequence of the
  \textit{Saccharomyces cerevisiae} genome with randomly generated
  gapped patterns.
  Top row: $6$ unit-length keywords, varying gap interval with a set of 50 and 100 patterns;
  Middle row: $6$ unit-length keywords, varying number of patterns with maximum gap 20 and 40;
  Bottom row: $2$ keywords, varying keyword length.}
\label{fig:benchq4}
\label{fig:benchq5}
\label{fig:benchq6}
\end{center}
\end{figure}

\section{Experimental results}\label{sec:experiments}

The proposed algorithms have been experimentally validated.
In particular, we compared the new algorithms \gqalgoC, \gqalgoD with the \textsc{d-pma} algorithm of~\cite{DBLP:conf/wea/HaapasaloSSS11}
and the \textsc{l-pma} algorithm of~\cite{DBLP:journals/tcs/BilleGVW12}.
The \gqalgoC and \gqalgoD have been implemented in the \textsf{C++} programming language
and compiled with the \texttt{GNU C++ Compiler 4.6}, using the
options \texttt{-O3}. The source code of the \textsc{d-pma} algorithm was kindly provided by the authors.
The test machine was a
3.00 GHz Intel Core 2 Quad Q9650 running Ubuntu 12.04 and running times were measured with the
\textsf{getrusage} function. The benchmarks consisted of searching for
a set of randomly generated gapped patterns in the DNA
sequence of $4,638,690$ base pairs of the \textit{Escherichia coli} genome
($\sigma=4$)\footnote{\url{http://corpus.canterbury.ac.nz/}} and in the protein sequence of $2,922,023$ symbols of the \textit{Saccharomyces
    cerevisiae} genome ($\sigma=20$)\footnote{\url{http://www.yeastgenome.org/}}.
The patterns were generated using the following procedure: given the
number $k$ of keywords, the length $l$ of each keyword and the maximum
length $b$ of a gap, we first randomly generate a sequence $g_1, g_2,
\ldots, g_{l-1}$ of $l-1$ gap lengths in the interval $[0,b]$; then,
we randomly sample a string of length $k\times l + \sum_{i=1}^{l-1} g_i$
from the text, and replace the substrings corresponding to the gaps
with their lengths.
Figures~\ref{fig:benchq1} and~\ref{fig:benchq4} show the experimental
results for the DNA and protein sequence, respectively. For each
sequence, we performed the following experiments:
\begin{enumerate}
\item (top row of Figures~\ref{fig:benchq1} and~\ref{fig:benchq4})
  searching a set of gapped patterns with $6$ keywords of unit length with a
  fixed number of patterns equal to $50$ and $100$, respectively, and
  such that the maximum gap varies between $5$ and $60$;
\item (middle row of Figures~\ref{fig:benchq1} and~\ref{fig:benchq4})
  searching a set of gapped patterns with $6$ keywords of unit length with a
  fixed maximum gap of $20$ and $40$, respectively, and such that the
  number of patterns varies between $25$ and $200$;
\item (bottom row of Figures~\ref{fig:benchq1} and~\ref{fig:benchq4})
searching a set of gapped patterns with $2$ keywords and a fixed maximum gap
of $20$ and such that the keyword length varies between $2$ and $6$.
In the benchmark to the left the number of patterns is calculated
using the formula $4 w / 2l$, where $l$ is the keyword length, so as
to fix the total number of symbols, i.e., $\plen$, to $4 w$ (i.e., $4$
words in our algorithm). In the one to the right the number of
patterns is fixed to $50$, so that $\plen$ increases as the keyword
length grows.
\end{enumerate}
We used a logarithmic scale on the y axis.
Note that the number of words used by our algorithm is equal to
$\ceil{6\times |\mathcal{P}| / w}$, so it is between $3$ and $19$ in
our experiments since $w=64$. Concerning the benchmark on DNA, the
experimental results show that the new algorithms are significantly
faster (up to 50 times) than the \textsc{d-pma} and \textsc{l-pma}
algorithms in the case of unit-length keywords (top and middle row).
in the case of arbitrary length keywords (bottom row), our algorithms
are significantly faster than $\textsc{d-pma}$ and $\textsc{l-pma}$ up
to keyword length $4$, while for longer keywords they have similar
performance. In the benchmark on the protein sequence the
\textsc{d-pma} and \textsc{l-pma} algorithms are considerably faster compared
to the case of DNA, which is expected since the average value of $\alpha$
and $\alpha'$ is inversely proportional to the alphabet size. Instead,
our algorithms exhibit a similar behaviour and are still faster than both
\textsc{d-pma} and \textsc{l-pma}.

The \gqalgoD algorithm is preferable if the text can be processed by
reading $w$ symbols at a time. This implies that, in the worst-case,
we report an occurrence of a pattern at position $i$ in the text only
after reading the symbols up to position $i+w-1$. This condition may
not be feasible for some applications. Otherwise, albeit slower, the
\gqalgoC algorithm is a good choice.
\section{Conclusions}

Motivated by a problem in computational biology, we have presented new
algorithms for the problem of \emph{multiple string matching} of
gapped patterns, where a gapped pattern is a sequence of strings such
that there is a gap of fixed length between each two consecutive
strings. The presented algorithms are based on dynamic programming and
bit-parallelism, and lie in a middle-ground among the existing
algorithms. In fact, their time complexity is close to the best
existing bound and, yet, they are also practical. We have also
assessed their performance with experiments and showed that they are
fast in practice and preferable if the strings in the patterns have
unit-length.

\section{Acknowledgments}

We thank the anonymous reviewers and Djamal Belazzougui for helpful comments.


\begin{footnotesize}
\bibliographystyle{plain}
\bibliography{gq}
\end{footnotesize}

\newpage

\appendix
\section{Proof of Theorem 1}

\setcounter{theorem}{0}

\begin{theorem}
Problem PMDBS is NP-hard in the strong sense.
\end{theorem}

\begin{proof}
Given an input $G = (V = \{v_1,\dots,v_n\},E = \{e_1,\dots,e_m\})$ to the Hamiltonian Path Problem, we construct the following instance $\mathcal{L}_G$ to Problem~\ref{problem-nphard} (see Fig.~\ref{fig-reduction} for an example).
\begin{itemize}
\item The universe $U$ consists of numbers $\{1,\dots,m\}$, which will be used to encode adjacencies, and numbers $\{m+1,\dots,n^2-m\}$, which will be used for padding, to ensure that all lists have the same length.
\item For every vertex $v_i \in V$, we have a list $L_i$ constructed as follows. Suppose the incident edges of $v_i$ are $e_{i_1}, e_{i_2}, \dots, e_{i_t}$, and say that the \emph{basic list of $L_i$} is the list $i_1,i_2,\dots,i_t$ padded (at the end) with $m-t$ new numbers from $\{m+1,\dots,n^2-m\}$, unused by any other list. List $L_i$ consists of $n$ concatenated copies of its basic list, so that $|L_i| = nm$.
\item We set $b = (n+1)m$ and $M = (2m - 1)(n-1) + m$.
\end{itemize}

We show that $G$ has a Hamiltonian path if and only if instance $\mathcal{L}_G$ admits a permutation $\pi$ of $1,\ldots,n$ such that the sum, over all lists $L^b$ in the $b$-mapping of $L_{\pi(1)},\dots,L_{\pi(n)}$, of the number of distinct elements in $L^b$ is at most $M$. Since the values of the integers in $U$ are bounded by a polynomial in the size of the lists $L_1,\dots,L_n$, this claim will entail the NP-hardness in the strong sense of Problem~PMDBS.

First, observe that from the choice of $b$ and of the lengths of lists $L_i$, for any permutation $\pi$ of $1,\dots,n$, the $b$-mapping $L_1^b,\dots,L_r^b$ of $L_{\pi(1)}$, $L_{\pi(2)}$, \ldots, $L_{\pi(n)}$ has a special form. Indeed, since $b = (n+1)m$, and the length of the lists $L_i$ is $nm$, we have that $r = \ceil{(n^2m)/((n+1)m)} = \ceil{n^2/(n+1)} = n$. It can be easily shown by induction that, for all $1 \leq j \leq n-1$, list $L_j^b$ consists of the last $(n-j+1)m$ integers in the list $L_{\pi(j)}$ followed by the first $jm$ integers from the list $L_{\pi(j+1)}$. List $L_{n}^b$ consists of the last $m$ integers of list $L_{\pi(n)}$. 

For the forward direction, let $P = v_{i_1},\dots,v_{i_n}$ be a Hamiltonian path of $G$. We show that the permutation $\pi$ of $1,\dots,n$ defined such that $\pi(j) = i_j$ satisfies the bound $M$. Let $L_1^b, L_2^b, \ldots, L_n^b$ be the $b$-mapping of $L_{\pi(1)},L_{\pi(2)}, \ldots, L_{\pi(n)}$. From the above observation, for all $1 \leq j \leq n-1$, the number of distinct integers in $L_j^b$ equals the number of distinct integers in $L_{\pi(j)}$, which is $m$, plus the number of distinct integers in $L_{\pi(j+1)}$, which is $m$, minus the number of integers shared between $L_{\pi(j)}$ and $L_{\pi(j+1)}$. Since $v_{\pi(j)}$ and $v_{\pi(j+1)}$ are connected by an edge, then the index of this edge appears in both $L_{\pi(j)}$ and $L_{\pi(j+1)}$, thus the number of distinct elements in $L_j^b$ is at most $2m - 1$. The claim is now clear, since $L_n^b$ consists of $m$ distinct integers.

For the backward implication, let $\pi$ be a permutation of $1,\dots,n$ such that the sum, over all lists $L^b$ in the $b$-mapping of $L_{\pi(1)},\dots,L_{\pi(n)}$, of the number of distinct elements in $L^b$ is at most $M$. We claim that the sequence $P = v_{\pi(1)},\dots,v_{\pi(n)}$ is a Hamiltonian path in $G$. Since $\pi$ is a permutation of $1,\dots,n$, we only have to show that for all $1 \leq i \leq n-1$, there is an edge between $v_{\pi(i)}$ and $v_{\pi(i+1)}$.

Let $L_1^b,L_2^b,\dots,L_n^b$ be the $b$-mapping of $L_{\pi(1)}$, $L_{\pi(2)}$, \ldots, $L_{\pi(n)}$. The fact that the number of distinct elements in the list $L_n^b$ is $m$ entails that the sum, over all $1 \leq j \leq n-1$, of the number of distinct elements in $L^b_j$ is at most $M-m = (2m-1)(n-1)$. 
For all $1 \leq j \leq n-1$, vertices $v_{\pi(j)}$ and $v_{\pi(j+1})$ have at most one edge incident to both of them (the edge connecting them), therefore, the number of distinct integers in each list $L_{j}^b$ is at least $2m-1$. From the above observation, for all $1 \leq j \leq n-1$, the number of distinct integers in each list $L_{j}^b$ is exactly $2m-1$.

Since the number of distinct integers in the list $L_{\pi(j)}$ is $m$ and the number of distinct integers in the list $L_{\pi(j+1)}$ is $m$, but the number of distinct integers in $L_j^b$ is at most $2m-1$, we have that lists $L_{\pi(j)}$ and $L_{\pi(j+1)}$ share at least one integer. We padded the basic lists of $L_{\pi(j)}$ and $L_{\pi(j+1)}$ with integers unique to them, thus the only integer shared by them must be the index of the edge incident to both $v_{\pi(j)}$ and $v_{\pi(j+1)}$. Such an edge connects $v_{\pi(j)}$ and $v_{\pi(j+1)}$, and thus $P$ is a path in $G$.
\end{proof}

\end{document}